\documentclass[nofootinbib,notitlepage,twocolumn,prl,noshowpacs,amsmath,amssymb]{revtex4-1}
\pdfoutput=1

\usepackage[titletoc,title]{appendix}
\usepackage{
	amssymb, 
	amsthm, 
	amsmath,
	mathtools,
	MnSymbol, 
	physics,
	bm,
	ifthen,
	etoolbox, 
	color, 
	phonetic,
	todonotes,
	qtree
}

\newcommand{\draftmode}{1}    
\newcommand{\notetoself}[1]{\ifnum \draftmode=1 \todo[inline,  backgroundcolor=blue!20!white]{#1} \fi}
\newcommand{\cuttext}[1]{\ifnum \draftmode=1 \todo[inline,bordercolor=black!5!white,backgroundcolor=black!5!white]{\color{black!70!white} #1} \fi}
\newcommand{\todoj}[1]{\ifnum \draftmode=1 \todo[inline,backgroundcolor=green!15!white]{#1} \fi}
\newcommand{\warntext}[1]{\ifnum \draftmode=1 \todo[inline, bordercolor=orange!30!white, backgroundcolor=orange!30!white]{#1} \else #1 \fi}

\let \projector \dyad

\newcommand{\Sys}{\ensuremath{\mathcal{S}}}
\newcommand{\Env}{\ensuremath{\mathcal{E}}}
\newcommand{\Frag}{\ensuremath{\mathcal{F}}}
\newcommand{\Fragp}{\ensuremath{{\mathcal{F}^\prime}}}
\newcommand{\FragBar}{{\ensuremath{{\overline{\mathcal{F}}}}}}

\newcommand{\Grag}{\ensuremath{\mathcal{G}}}
\newcommand{\Hrag}{\ensuremath{\mathcal{J}}}

\newcommand{\Glob}{\ensuremath{\mathcal{H}}}

\newcommand{\mcA}{\ensuremath{\mathcal{A}}}
\newcommand{\mcB}{\ensuremath{\mathcal{B}}}
\newcommand{\mcC}{\ensuremath{\mathcal{C}}}

\newcommand{\Pip}{\Pi}
\newcommand{\Disj}[1]{\hat{#1}}

\newcommand{\Obs}{\ensuremath{\Omega}}
\newcommand{\obs}{\omega}

\newcommand{\rarw}{\odot}
\newcommand{\larw}{\otimes}

\newtheorem*{mainthm}{Main Result} 

\newtheorem*{cor}{Corollary} 
\theoremstyle{remark}
\newtheorem*{remark}{Remark} 
\theoremstyle{definition}
\newtheorem*{definition}{Definition}

\renewcommand{\draftmode}{0}

\begin{document}

\title{Classical Branch Structure from Spatial Redundancy in a Many-Body Wavefunction${}^1$}
\date{\today}
\author{C.~Jess~Riedel}\email{jessriedel@gmail.com}
\affiliation{Perimeter Institute for Theoretical Physics, Waterloo, Ontario N2L 2Y5, Canada}

\begin{abstract}

When the wavefunction of a large quantum system unitarily evolves away from a low-entropy initial state, there is strong circumstantial evidence it develops ``branches'': a decomposition into orthogonal components that is indistinguishable from the corresponding incoherent mixture with feasible observations.  Is this decomposition unique?  Must the number of branches increase with time?  These questions are hard to answer because there is no formal definition of branches, and most intuition is based on toy models with arbitrarily preferred degrees of freedom. Here, assuming only the tensor structure associated with spatial locality, I show that branch decompositions are highly constrained just by the requirement that they exhibit redundant local records. The set of all redundantly recorded observables induces a preferred decomposition into simultaneous eigenstates unless their records are highly extended and delicately overlapping, as exemplified by the Shor error-correcting code. A maximum length scale for records is enough to guarantee uniqueness. Speculatively, objective branch decompositions may speed up numerical simulations of nonstationary many-body states, illuminate the thermalization of closed systems, and demote measurement from fundamental primitive in the quantum formalism.
\end{abstract}

\maketitle

Given the wavefunction $\ket{\psi}$ of a many-body system at a given time, we seek to identify a unique decomposition into orthogonal components, 
\begin{align}
\label{eq:branches}
\ket{\psi} = \sum_i \ket{\psi_i},
\end{align}
that have effectively ``collapsed'' in the intuitive sense that their coherent superposition can't be distinguished from the incoherent mixture $\rho = \sum_i \dyad{\psi_i}$ by feasible observations.  This decomposition should be as general and abstract as possible, without a-priori reference to a preferred observer, a preferred apparatus, a preferred set of observables, or a preferred system distinguished from a remaining environment; all these should emerge.  

Without additional structure, every state $\ket{\psi}\in \Glob$ in the Hilbert space is equivalent.  The minimal ingredient we choose to assume is a division of the many-body system into \emph{microscopic} sites (e.g., qubits), which mathematically takes the form of tensor-product structure:
\begin{gather}
\label{eq:lat}
\Glob = \bigotimes_n \Env^{(n)}
\end{gather}
We take \eqref{eq:lat} as a primitive that is ultimately grounded in spatial locality. The associated multipartite entanglement \cite{amico2008entanglement} in $\ket{\psi}$ provides a rich foundation.

Our guiding intuition is that when macroscopically distinct alternatives decohere \cite{zurek1981pointer,*zurek1982environment-induced,*zurek2003decoherence,anastopoulos2002frequently,joos2003decoherence,schlosshauer2008decoherence}, redundant records about the outcome (defined precisely below) are produced through the phenomenon of quantum Darwinism \cite{zurek2000einselection,*zurek2009quantum,ollivier2004objective,*ollivier2005environment,blume-kohout2006quantum,bennett2008second,brandao2015generic,horodecki2015quantum,riedel2016objective}.  In retrospect this is plainly true in the special case of laboratory measurements, where abundant classical correlations are evident in, e.g., the measuring apparatus itself, in the circuits of the electronic readout, in the photons emitted by a display, and in the brains of nearby observers.  Much more commonly, and less obviously, correlated records are naturally and prolifically produced in non-anthropocentric mesoscopic processes, such as when quantum fluctuations are amplified by classically chaotic systems \cite{zurek1994decoherence,*zurek1995quantum,*zurek1999why,elze1995vacuum-induced,zurek1998decoherence-a,pattanayak1999lyapunov,monteoliva2000decoherence,jalabert2001environment-independent} and subsequently decohered by ubiquitous environments like scattered photons \cite{riedel2010quantum,* riedel2011redundant,korbicz2014objectivity}.  
Our strategy is to identify wavefunction branches, at a fixed time, with the multipartite entanglement structure associated to the records generically generated in the wake of these dynamical processes.
We expect records in many different locations, but these need not be microscopically local, so we will look for records to exist in spatial regions -- subsets of the entire lattice \eqref{eq:lat}.

As shown rigorously in this letter, a set of recorded observables induces an objectively preferred decomposition of the wavefunction into branches \eqref{eq:branches} -- each a simultaneous eigenstate of the entire set -- so long as no \emph{two} records of one observable, taken together, spatially overlap \emph{all} records of another.  
Redundancy alone, no matter how large, is not sufficient to guarantee objectivity, but all counterexamples necessarily feature many unnaturally elongated and delicate records, as exhibited by the Shor error-correcting code \cite{shor1995scheme}.  
In fact, the set of \emph{all} observables recorded redundantly on regions bounded by any particular length scale induces a single preferred decomposition of the wavefunction into branches.  This is shown without appeal to arbitrarily preferred macroscopic degrees of freedom, and without breaking any symmetries of the lattice \eqref{eq:lat}, e.g., invariance under translations, rotations, and reflections.

Consider any observable $\Obs^\Frag$ local to some region $\Frag = \bigotimes_{n \in \Frag} \Env^{(n)}$ of the larger Hilbert space $\Glob = \Frag \otimes \FragBar$ containing $\ket{\psi}$.  Let the eigen-decomposition be
\begin{gather}
\Obs^\Frag = \sum_i \obs_i \Pip^\Frag_i,\qquad  \obs_i \in \mathbb{R}, \qquad \Pip^\Frag_i \Pip^\Frag_j = \delta_{ij} \Pip^\Frag_i
\end{gather}
where the $\Pip^\Frag_i = (\Pip^\Frag_i)^2 = (\Pip^\Frag_i)^\dagger$ are orthogonal projectors onto the (generally degenerate) subspaces of $\Frag$ associated with the distinct eigenvalues $\obs_i$, acting trivially on $\FragBar$.

\begin{definition}	
	We say a local observable $\Obs^{\Frag}$ \emph{records} another local observable $\Obs^{\Fragp}$ on a disjoint region $\Fragp$ when, for each $i$,
	\begin{align}
	\label{eq:record}
	\Pip_i^{\Frag} \ket{\psi} = \Pip_i^{\Fragp} \ket{\psi}. 
	\end{align}
	This is a symmetric relation,  naturally extending to a collection $\Obs \equiv \{\Obs^\Frag,\Obs^{\Frag'},\Obs^{\Frag''},\dots\}$ of two or more local observables, on disjoint regions $\{\Frag,\Frag',\Frag'',\dots\}$, recording each other.  We discuss $\Obs$ collectively as a \emph{recorded observable}, referring to $\Obs^\Frag$ as a \emph{record} of $\Obs$ on the region $\Frag$, and the number of records $\vert \Obs \vert$ as the \emph{redundancy} of $\Obs$. Finally, we define the unnormalized \emph{branch} corresponding to $i$ 
	as $\ket{\psi_i} \equiv \Pip_i^{\Frag} \ket{\psi} = \Pip_i^{\Fragp} \ket{\psi} = \Pip_i^{\Frag''} \ket{\psi} = \cdots$.
\end{definition}
\begin{remark}
	Note that $\Obs^{\Frag}$ records $\Obs^{\Fragp}$ if and only if $\Pip_j^{\Frag}  \rho^{\Frag}_{\Fragp:i} \Pip_j^{\Frag}  = \delta_{ij} \rho^{\Frag}_{\Fragp:i}$, where $\rho^{\Frag}_{\Fragp:i} \equiv \Tr_{\FragBar}[ \Pip_i^{\Fragp}\projector{\psi}\Pip_i^{\Fragp}]$ is the state local to $\Frag$ corresponding to the eigenvalue $\obs_{\Fragp:i}$ of $\Obs^\Fragp_i$.  
	Therefore, a local observer can make a measurement on $\Frag$ to infer the value of $\Obs^{\Fragp}$, and similarly for $\Fragp$ and $\Obs^{\Frag}$. In other words, each branch $\ket{\psi_i}$ lives in its own subspace of the local\footnote{This \emph{local} orthogonality of quantum states \cite{horodecki1998entanglement,riedel2013local,riedel2016objective} induces a corresponding constraint on the statistics of measurement outcomes \cite{fritz2013local,sainz2014exploring}.} Hilbert spaces $\Frag$ and $\Fragp$.  Nothing here depends on the actual eigenvalues since they only label the different eigenspaces. In this sense, the object being recorded is a local subalgebra of block-diagonal matrices rather than an observable per se.
\end{remark}

A salient characteristic of macroscopic observables, whether or not associated with the result of laboratory measurements, is that they are recorded with very high redundancy, satisfying \eqref{eq:record} to high accuracy.\footnote{Perfect (von-Neumann \cite{vonneumann1932mathematische}) measurement by a spatially localized apparatus is a measure-zero idealization that real-world experiments can approach with arbitrary precision \cite{dowker1996consistent,wiseman2014quantum}, with relativistic limitations that become exponentially small on scales larger than the relevant Compton wavelength \cite{haag1996local}.}   (More eventually needs to be said about imperfect records and quantifying redundancy, but ultimately this will be an approximate notion like thermodynamic irreversibility, which becomes unambiguous in a large-$N$ limit.)  Our goal is to determine under what conditions there exists a preferred decomposition of the wavefunction into branches that are simultaneous eigenstates of \emph{all} redundantly recorded observables, thereby assigning the branches to the outcomes of performed measurements.

Consider a set of several redundantly recorded observables $\{\Obs_a,\Obs_b,\Obs_c,\dots \}$ whose corresponding eigenvalues are labeled by $i$, $j$, $k$, etc.  In agreement with our real-world expectations, we require that there are records in multiple places of different observables, but do not require that any single region contains a record of all such observables.\footnote{If there is a multi-valued observable for which different eigenvalues are distinguished by records on different regions, it can be decomposed into multiple binary observables each of whose records are unambiguously in a particular region.}   Nonetheless, the records of different macroscopic observables may generally be on overlapping regions, so that they are not guaranteed to commute.  (That is, if $\Obs_a$ is recorded on disjoint regions $\Frag$ and $\Frag'$, and $\Obs_b$ is recorded on disjoint regions $\Grag$ and $\Grag'$, $\Frag$ may still overlap with $\Grag$. See Fig.~\ref{fig:figure}.)  Given this, we would like to determine under what conditions they are all mutually compatible, as expected for classical objectivity.

\begin{definition}
	Suppose $\{\Obs_a=\{\Obs_a^\Frag,\Obs_a^\Fragp,\dots\}\}$ is a set of redundantly recorded observables.  We say the $\Obs_a$ are \emph{compatible} on $\ket{\psi}$ if there exists a decomposition
	\begin{align}
	\label{eq:jointdecomp}
	\ket{\psi} = \sum_{i,j,k,\ldots} \ket{\psi_{i,j,k,\ldots}}
	\end{align}
	where 
	the unnormalized $\ket{\psi_{i,j,k,\ldots}}$ are simultaneous eigenstates of all records in $\{\Obs_a\}$, i.e.,
	\begin{align}
	\label{eq:jointeigen}
	\Obs_a^\Frag \ket{\psi_{i,j,k,\ldots}} = \obs_{a:i} \ket{\psi_{i,j,k,\ldots}}
	\end{align}
	for all $a$, for all $\Obs_a^\Frag\in\Obs_a$, and for all $i$ indexing the real eigenvalues $\obs_{a:i}$ of $\Obs_a^\Frag$.  We call the $\ket{\psi_{i,j,k,\ldots}}$ the \emph{branches} of the joint decomposition. 

\end{definition}

If a set of recorded observables $\{\Obs_a\}$ is compatible on $\ket{\psi}$, it follows that the joint branch decomposition \eqref{eq:jointdecomp} is orthogonal and unique [since \eqref{eq:jointeigen} is equivalent to $\Pip^{\Frag}_{a:i'} \ket{\psi_{i,j,k,\ldots}} = \delta_{i,i'}\ket{\psi_{i,j,k,\ldots}}$], and the branches span a subspace on which all records commute.

Joint branch structure recovers the Everettian intuition that local records can inform localized observers.  It also suggests the unambiguous definition of the coarse-grained branches
$\ket{\psi_{a:i}} \equiv \sum_{j,k,\cdots} \ket{\psi_{i,j,k,\ldots}}$,  $\ket{\psi_{a:i,b:j}} \equiv \sum_{k,\cdots} \ket{\psi_{i,j,k,\ldots}}$, 
etc., and implies the corresponding coarse-graining relationships
$\ket{\psi_{a:i}} = \sum_{j} \ket{\psi_{a:i,b:j}} = \sum_{j,k} \ket{\psi_{a:i,b:j,c:k}}$,
etc. The partially coarse-grained branches are eigenstates of the operators that have not been coarse-grained over.

One can see that compatibility of recorded observables is not trivial: the Bell state
\begin{align}
\label{eq:bell}
\ket{\Phi^{+}} \propto \ket{\uparrow} \ket{\uparrow} + \ket{\downarrow}\ket{\downarrow} = \ket{\rarw}\ket{\rarw} + \ket{\larw} \ket{\larw}
\end{align}
with $\ket{\rarw} \equiv (\ket{\uparrow}+\ket{\downarrow})/\sqrt{2}$ and $\ket{\larw} \equiv (\ket{\uparrow}-\ket{\downarrow})/\sqrt{2}$,
features two observables, $\Obs_{\uparrow, \downarrow}$ and $\Obs_{ \rarw,\larw}$, that are recorded locally twice (once on each qubit) yet are incompatible.\footnote{Note they are incompatible even if auxiliary dimensions are appended to the Hilbert space of the Bell qubits.}

In fact, two observables can each be recorded with \emph{arbitrarily} large redundancy yet be grossly incompatible -- corresponding to noncommuting observables.  An example of this is provided by the generalized Shor code \cite{shor1995scheme}, (a class of) states used to represent quantum information in error-correctable form:
\begin{gather}\begin{split}
\label{eq:shortshor}
\ket{\psi} = \ket{\xi_+} + \ket{\xi_-},\qquad \ket{\xi_\pm}\equiv \left[\ket{0}^{\otimes M'} \pm  \ket{1}^{\otimes M'}\right]^{\otimes M}.
\end{split}\end{gather}
The first incompatible observable is $\Obs_{\pm}$, which corresponds to the branch decomposition above, and which is recorded with redundancy $M$.  The second is $\Obs_{0,1}$, which corresponds to the decomposition $\ket{\psi} = \ket{\chi_\mathrm{0}} + \ket{\chi_\mathrm{1}}$, and which is recorded with redundancy $M'$.  Here,
$\ket{\chi_{r}} = \sum_{\vec{s} \in Z_r}  \bigotimes_{m=1}^{M'} [ \ket{s_m}^{\otimes M}]$ for $r=0,1$,
where $\vec{s} = (s_1,\ldots,s_{M'})$ with $s_m = 0,1$ is a vector of bits, and where $Z_0$ ($Z_1$) denotes the set of such vectors with even (odd) parity.  The record structure is illustrated in Fig.~\ref{fig:figure}(a) and further detail can be found in the Appendix.

Therefore, additional assumptions beyond mere redundancy will be required to identify the preferred macroscopic observables inducing branch structure. We now introduce an important (but initially obscure) asymmetric binary relation on a set of recorded observables and prove it necessarily holds for some pairs if the set is not compatible;  otherwise, they induce a joint branch decomposition.\footnote{Sufficient conditions are discussed in the Appendix.}

\begin{definition}
	Suppose two observables $\Obs_a$ and $\Obs_b$ are redundantly recorded on $\ket{\psi}$. Then we say $\Obs_a$ \emph{pair-covers} $\Obs_b$ if there is at least one pair of records $\Obs_a^\Frag,\Obs_a^\Fragp \in \Obs_a$ such that, for every $\Obs_b^\Grag \in \Obs_b$, the region $\Grag$ spatially overlaps with $\Frag$ or $\Fragp$ (or both). Equivalently, $\Obs_a$ does \emph{not} pair-cover $\Obs_b$ if, for every pair of records $\Obs_a^\Frag,\Obs_a^\Fragp \in \Obs_a$, there exists a record $\Obs_b^{{\Grag}} \in \Obs_b$ such that $\Grag$ is disjoint from both $\Frag$ and $\Fragp$.
	[See Fig.~\ref{fig:figure}(b).]
\end{definition}

Given the very many physical records that exist about macroscopic observables, we do not expect that a pair of accessible records for one observable to spatially overlap with \emph{all} records of another.  Even if an observable has some spurious, highly diffuse records in addition to the localized ones that are feasibly accessible to observers, the modified recorded observable formed by simply dropping the diffuse records should avoid pair-covering other macroscopic observables.  This procedure only fails if most or all of the records are extensively overlapping in this way.  Indeed, the Shor code exemplifies this; its two incompatible recorded observables pair-cover each other regardless of how many records are dropped, since each record of one observable covers all records of the other.  [See Fig.~\ref{fig:figure}(a).]  For large redundancy, the records must become arbitrarily extended in space.

\newcommand{\maintheoremstatement}{Suppose $\{\Obs_a=\{\Obs_a^\Frag,\Obs_a^\Fragp,\dots\}\}$ is a set of redundantly recorded observables for $\ket{\psi}$.  If none of the recorded observables pair-covers another, then they are all compatible, and so define a joint branch decomposition of simultaneous eigenstates of all records.}

\begin{mainthm}\maintheoremstatement\end{mainthm}

\begin{proof}
(Sketch.) The strategy is to show that an arbitrary product of record projectors ($\Pip^{\Frag}_{a:i}\Pip^{\Grag}_{b:j}\cdots$) acting on $\ket{\psi}$, as in \eqref{eq:jointdecomp}, is independent of both the order of the $\Pi$'s and of the particular choice of $\Obs_a^\Frag \in \Obs_a$, $\Obs_b^\Grag \in \Obs_b$, etc. The proof is by induction on the number of projectors in the product, starting with two: $\Pip^{\Frag}_{a:i}\Pip^{\Grag}_{b:j}\ket{\psi} = \Pip^{\Grag}_{b:j}\Pip^{\Frag}_{a:i}\ket{\psi} = \Pip^{\Grag'}_{b:j}\Pip^{\Fragp}_{a:i}\ket{\psi}$. All steps are elementary, consisting of repeated application of the definition of local records (i.e., $\Pip^\Frag_{a:i}\ket{\psi} = \Pip^{\Frag'}_{a:i}\ket{\psi}$ for all $\Obs_a^\Frag,\Obs_a^\Fragp \in \Obs_a$), and the lack of pair-covering (i.e., $[\Pip^\Frag_{a:i},\Pip^{\Disj{\Grag}}_{b:j}t]=0=[ \Pip^{\Frag'}_{a:i},\Pip^{\Disj{\Grag}}_{b:j}]$ for all $\Obs_a^\Frag,\Obs_a^\Fragp \in \Obs_a$ and for some choice $\Obs_b^{\Disj{\Grag}} \in \Obs_b$ with $\Disj{\Grag} = \Disj{\Grag}(\Frag,\Frag')$). 
See the Appendix for details.
\end{proof}

This result gives evidence that our intuition about records may be enough to fully constrain the branch structure of a many-body wavefunction.  However, it does not necessarily single out a unique decomposition.  An ideal criterion for classical observables could be checked on \emph{individual} recorded observables yet guarantee mutual compatibility, thereby identifying a single maximal set.

Note that such a criterion must make reference to something besides scale-invariant properties of the recording regions.  Given an arbitrary set of regions on which some observable is recorded, an incompatible observable can be recorded on a dilated but otherwise identical set of regions.  A state fulfilling this is
\begin{align}
\label{eq:dilatedpsi}
\sum_{\pm} \left[(\ket{0}_\Grag \ket{0}_{\Grag'}\cdots) \pm (\ket{1}_\Grag \ket{1}_{\Grag'}\cdots)\right]( \ket{\pm}_{\Frag'}\ket{\pm}_{\Frag''}\cdots),\,\,\,
\end{align}
where $\Frag = ({\Grag} \otimes {\Grag'} \otimes \cdots)$ is a region in which \emph{one} record of $\Omega_a$ and \emph{all} records of $\Omega_b$ are inscribed.  This is illustrated in Fig.~\ref{fig:figure}(c). 
In other words, if we know only the regions on which putatively classical information is redundantly recorded, it is always possible that incompatible but redundantly recorded information hides at very small or very large length scales.

The following corollary assumes a preferred length scale to state a criterion that can be checked on individual recorded observables.  It is not fully satisfactory as a fundamental criterion for objective branch structure, but it illustrates the form that such a criterion could take.

\begin{figure} [b]
	\centering 
	\newcommand{\pbwidthfactor}{0.98}
	\includegraphics[width=\pbwidthfactor\columnwidth]{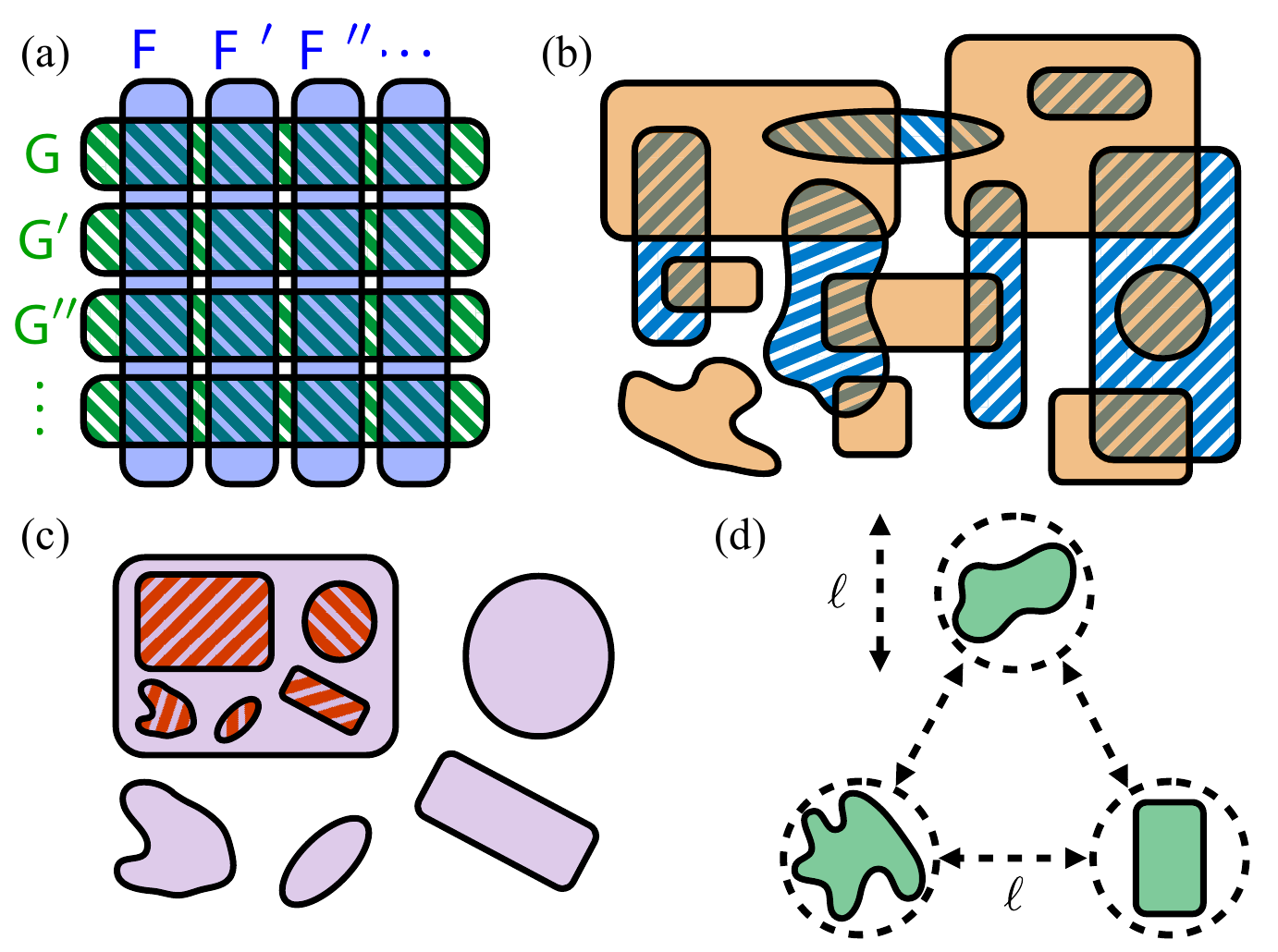}
	\caption{
		Spatially disjoint regions with the same coloring (e.g., the solid blue regions $\Frag, \Fragp, \ldots$) denote different records for the same observable (e.g., $\Obs_a = \{\Obs_a^\Frag,\Obs_a^\Fragp,\ldots\}$). (a) The spatial record structure of the Shor-code family of states, which can exhibit arbitrary redundancy (in this case four-fold) for two incompatible observables. (b) The solid orange observable pair-covers the hashed blue observable because the top two orange records overlap all blue records.  However, if one of the top two orange records is dropped, then neither observable pair-covers the other, and hence both are compatible, despite many overlaps of individual records.  (c) Any spatially bounded set of records can be contained inside a single record of a sufficiently dilated but otherwise identical set of records for an incompatible observable; such a state is given in Eq.~\eqref{eq:dilatedpsi}.  
		(d) Any observable with records satisfying the hypothesis of the Corollary for some length $\ell$ cannot pair-cover, or be pair-covered by, any other such observable.
	}
	\label{fig:figure}
\end{figure}

\newcommand{\corollaryDstatement}{Fix a characteristic spatial distance $\ell$ and consider the set of all recorded observables $\{\Obs_a\}$ on $\ket{\psi}$ satisfying the following requirement: each $\Obs_a$ is recorded on at least $3$ regions, with each region fitting in a sphere of radius $\ell$ and pairwise separated by the distance $\ell$. [See Fig.~\ref{fig:figure}(d)].  Then none of the $\Obs_a$ pair-covers another, and they are all therefore compatible and define a joint branch decomposition.}

\begin{cor}\label{corsphere}\corollaryDstatement\end{cor}

\begin{remark}
	This bound is tight in the sense that incompatible observables each with two such records can exist.  The special role of the number three in this bound is fundamental, and is essentially the same as in the triorthogonal decomposition theorem \cite{elby1994triorthogonal}.  In both cases, we are able to rule out quantum effects in information that is distributed over more than two subsystems because of the monogamy of entanglement \cite{terhal2001family,coffman2000distributed}.
	
	This corollary is the first place we have associated the lattice with a notion of distance (or even topology), and it only functions to ensure that regions are disjoint.  The distance $\ell$ might be motivated by a fundamental correlation scale of the state $\ket{\psi}$, or the maximum distance over which realistic observers can make measurements.  (Poetically, a macroscopic wavefunction at any given time has a unique branch decomposition generated by \emph{all} observables that could be recorded in several human brains, $\sim\!\! 20$ cm.) Note that the mere existence of some records that become diffused over a distance larger than $\ell$ does not interfere with applying the Corollary to the modified recorded observable formed by dropping the superfluous diffuse records.
\end{remark}

\emph{Discussion}. It seems very unlikely that the recorded observables corresponding to traditional laboratory measurements would pair-cover one another, by virtue of the millions \cite{riedel2010quantum} of localized records distributed over macroscopic distances, so they are expected to generate a joint branch decomposition in the wavefunction of the universe.  More generally, we expect the same when classically chaotic systems amplify quantum fluctuations, which then decohere \cite{zurek1994decoherence,elze1995vacuum-induced,zurek1995quantum,zurek1999why,zurek1998decoherence-a,pattanayak1999lyapunov,monteoliva2000decoherence,jalabert2001environment-independent}, without any involvement of observers or laboratory equipment.   In contrast to idealizations that assume that different record-holding regions are approximately separable \cite{korbicz2014objectivity}, the above construction is not stymied by the presence of stray entanglement, an unavoidable aspect of the real world \cite{zurek2013wave-packet}.

That said, there are important limitations that remain to be addressed.   We have not shown that a branch decomposition is stable in the presence of small errors or imperfect records.  We have also not resolved how the decomposition, which is induced by the locally causal production of records, would transform under relativistic boosts, nor how it would be defined if there is no Cauchy surface.  The preferred tensor structure \eqref{eq:lat} is justified by the universal nature of spatial locality (see also \cite{cao2016space,cotler2017locality}), but this structure is not applicable on scales smaller than the Compton wavelength of a relativistic quantum field \cite{fredenhagen1985remark,summers1990independence,haag1996local,zych2010entanglement}. Very importantly, the hypothesis of the Corollary relies on an unexplained length scale and does not obviously agree with intuition in all cases; a fundamental uniqueness theorem is lacking.

Ideally, objective branch structure would be built up from a Lorentz and scale invariant condition (cf.\ Refs.~\cite{kent2014solution,*kent2015lorentzian,bousso2012multiverse}). 
This might be based on a preferred length scale or inertial frames extracted from the state $\ket{\psi}$, or it might appeal to other principles, such as the information redundantly recorded in \emph{most} collections of lattice sites \cite{brandao2015generic,fernandoPC}. In any case, the decomposition will only be convincing if it is simple and rigorously recovers all intuition about the evolution of macroscopic observables.
Of course, for some states the decomposition may be trivial (just one branch) but insofar as it is defined by properties that the macroscopic classical world is expected to obey, the induced branches would be nontrivial and objectively exist ``out there in the real world'' -- they would not be just a useful structure relative to a particular observer.  

The production of records during macroscopic amplification is a thermodynamically irreversible process; in principle, it is always possible to conduct quantum experiments in a perfectly sealed laboratory and have the resulting outcomes ``recohered'' by an external agent with sufficiently powerful abilities.  Therefore, we do not expect branching to occur at an exact moment in time, but rather to emerge in a large-$N$ limit.  (Certainly, the production of only three records, as in a GHZ state \cite{amico2008entanglement}, is not enough to ensure persistent objectivity of the recorded observable.)  So it is likely that some distinguishability metric between candidate branches could usefully quantify the permanence of branching.  One possibility is simply the amount of redundancy, which is somewhat analogous to the Hamming distance between alternative branches.  Perhaps more compelling is (a computable approximation to) the logical depth \cite{bennett1995logical} or quantum circuit complexity \cite{yao1993quantum, aaronson2004multilinear} between branches; with careful preparation, matter interferometers successfully interfere different configurations of $10^4$ nucleons \cite{gerlich2011quantum} (and hence, in a sense, $10^4$ records), but this would be infeasible if the two configurations were well scrambled\footnote{I thank Adam Brown and Leonard Susskind for this example.} \cite{sekino2008fast}.  Given a sufficient threshold, we expect branches to divide, but not recombine, under time evolution. For time-homogeneous systems, branches presumably divide into sub-branches at a regular rate, leading to a total number of branches that increases exponentially in time.

For the sake of argument, suppose we have identified a satisfying (though possibly laborious) procedure for decomposing the wavefunction at any given time into branches of simultaneous eigenstates of preferred observables.\footnote{Such a procedure can be applied recursively to each branch in the decomposition, allowing for \emph{branch-dependent} structure \cite{gell-mann1990quantum,mueller2007branch}, such as when the observable measured by an experiment is conditioned on the outcome of a previous one.}   Given this, we can investigate how the number and type of branches at different times relate to each other. Many compelling questions could be investigated:  What is the behavior of the entropy defined by the spectrum of squared branch weights ($\{ \vert \ket{\psi_{i,j,k,\dots}}\vert^2\}$) \cite{brun1999entropy}, and is it related to the Kolmogorov-Sinai entropy of the macroscopic degrees of freedom \cite{zurek1998decoherence-a}?  At what branch-distance threshold is branch formation irreversible? Finite-dimensional systems allow for at most a finite number of orthogonal branches \cite{diosi1995maximum,dowker1995properties,dowker1996consistent}; when does branching halt, and what does the transition look like?  Since thermalized systems are characterized by a lack of redundancy -- disjoint local measurements inside a uniform-temperature oven are completely uncorrelated with each other -- can the destruction of records connect the dissolution of branch structure with the thermalization process itself?

If a computationally efficient method for identifying branches for a given state $\ket{\psi}$ could be found, it would enable simulations of any non-stationary many-body systems whose failure to be compactly described by a tensor network is due to an exponential proliferation of branches -- and, hence, long-range entanglement.  Such entanglement builds up, for instance, when local excitations scatter into superpositions of different out-states. Crucially, an $N$-point function can be calculated by \emph{sampling} the branches for observables recorded more than $N$ times: $\matrixelement{\psi}{\,\mathcal{O}_1 \cdots\, \mathcal{O}_N} {\psi} = \sum_{i,j,\dots} \matrixelement{\psi_{a:i,b:j,\dots}}{\,\mathcal{O}_1 \cdots\, \mathcal{O}_N} {\psi_{a:i,b:j,\dots}}$, because there exists choices of $\Obs^\Frag_a, \Obs^\Grag_b, \dots,$ such that $0=[\Pip^\Frag_{a:i},\mathcal{O}_n]=[\Pip^\Grag_{b:j},\mathcal{O}_n]=\dots$, for all $n$. 
Here, the sum need only include enough choices of $(i,j,\dots)$ to ensure a small error, which scales polynomially with the desired accuracy and is independent of the number of branches for finite variance. 
As the total number of branches increases exponentially with time evolution, the number that need to be simulated can be held constant; some branches would be retained with probability proportional to their norm squared, and the rest ``pruned''.

In principle, an objective branch decomposition of the wavefunction of the universe at every moment in time could reduce quantum theory to a classical stochastic theory -- without invocation of observers or measurements as primitive concepts -- in the following sense: It would provide a well-behaved probability distribution over different outcomes, and for each outcome it would specify a preferred set of observables and their values (while remaining appropriately silent on the values of incompatible observables).  These observables would follow quasiclassical trajectories over timescales on which the conditions for Ehrenfest's theorem hold.  It would thus convert the ad-hoc operational procedure by which quantum mechanics is applied \cite{dowker1996consistent,bassi2000decoherent,okon2014measurements,wallace2016what} into a formal calculus.

\begin{acknowledgments}
I thank Scott Aaronson, Charles Bennett, Adam Brown, Todd Brun, Josh Combes, Martin Ganahl, Steve Giddings, Daniel Gottesman, James Hartle, Adrian Kent, Siddharth Muthukrishnan, Don Page, Fernando Pastawski, John Preskill, Leonard Susskind, Guifre Vidal, Tian Wang, I-Sheng Yang, Wojciech Zurek, and Michael Zwolak for inspiration, illuminating discussion, and feedback. Research at Perimeter Institute is supported by the Government of Canada through Industry Canada and by the Province of Ontario through the Ministry of Economic Development and Innovation.
\end{acknowledgments}
	
\bibliographystyle{apsrev4-1}
\bibliography{zotriedel,mylocalbib}

\appendix

\section{Appendix A: Proofs}
\label{sec:mainproof}

Here we restate and prove the Main Result and Corollary.\footnote{Below and in the main Letter,  we assume the following slightly ambiguous convention: For any three recorded observables $\Obs_a$, $\Obs_b$, and $\Obs_c$, they are recorded respectively by the regions $\{\Frag, \Frag',\cdots\}$, $\{\Grag,\Grag',\cdots\}$, and $\{\mathcal{I},\mathcal{I}',\cdots\}$.  Likewise, $i$, $j$, and $k$ index the eigenvalues of $\Obs_a$, $\Obs_b$, and $\Obs_c$, respectively. This avoid an additional layer of subscripts, e.g., $\Frag_a^\prime$, $\Frag_b$, $i_b$, etc.}

\begin{mainthm}\maintheoremstatement\end{mainthm} 

\begin{proof}
	We will show that a product of record projectors ($\Pip$'s), when acting on the state $\ket{\psi}$, is independent of both the order of the projectors and of the choice of recording regions for each particular recorded observable.  That is, we show
	\begin{align}
	\label{eq:canrelabel}
	\left(\Pip^\Frag_{a:i}\Pip^\Grag_{b:j}\cdots \right)\ket{\psi} = \left(\Pip^{\Frag'}_{a:i}\Pip^{\Grag'}_{b:j}\cdots  \right)\ket{\psi} 
	\end{align}
	for any $\Obs_a^\Frag,\Obs_a^{\Frag'} \in \Obs_a$, and $\Obs_b^\Grag,\Obs_b^{\Grag'} \in \Obs_b$, and so on, and that furthermore this object is unchanged when the $\Pip$'s are commuted with one another.
	
	The strategy used here is simple but laborious, essentially arising from repeated application of the two basic identities defining local records,
	\begin{align}\begin{split}
	\forall \Obs_a^\Frag,\Obs_a^{\Frag'} \!\in\! \Obs_a, \forall i, \qquad \Pip^\Frag_{a:i}\ket{\psi} = \Pip^{\Frag'}_{a:i}\ket{\psi},
	\end{split}\end{align}
	and the failure to pair-cover
	\begin{align}\begin{split}
	\forall \Obs_a^\Frag,\Obs_a^{\Frag'} \!\in\! \Obs_a ,\, \exists &\Obs_b^{\Disj{\Grag}} \!\in\! \Obs_b, \forall i, \forall j, \\
	&\left[ \Pip^\Frag_{a:i},\Pip^{\Disj{\Grag}}_{b:j}\right]=0=\left[ \Pip^{\Frag'}_{a:i},\Pip^{\Disj{\Grag}}_{b:j}\right],
	\end{split}\end{align}
	where $\Disj{\Grag} = \Disj{\Grag}(\Frag,\Frag')$.

	The proof is by induction.  First, we need to demonstrate that if we have just a pair of local record projectors acting on the state, $\Pip^\Frag_{a:i}\Pip^\Grag_{b:j}\ket{\psi}$, the object is unchanged if we send $\Frag \to \Frag'$, $\Grag \to \Grag'$ for arbitrary new $\Obs_a^{\Frag'} \in \Obs_a,\Obs_b^{\Grag'} \in \Obs_b$.  That is, we must prove
	\begin{align}
	\Pip^\Frag_{a:i}\Pip^\Grag_{b:j}\ket{\psi}
	&=  \Pip^{\Frag'}_{a:i} \Pip^{\Grag'}_{b:j}  \ket{\psi}.
	\end{align}
	This is shown by repeated application of our two basic identities:
	\begin{align}\begin{split}
	\Pip^\Frag_{a:i}\Pip^\Grag_{b:j}\ket{\psi} &= \Pip^\Frag_{a:i}\Pip^{\Disj{\Grag}}_{b:j}\ket{\psi} \\
	&= \Pip^{\Disj{\Grag}}_{b:j} \Pip^\Frag_{a:i}\ket{\psi} \\
	&= \Pip^{\Disj{\Grag}}_{b:j} \Pip^{\Frag'}_{a:i} \ket{\psi} \\
	&=  \Pip^{\Frag'}_{a:i} \Pip^{\Disj{\Grag}}_{b:j} \ket{\psi} \\
	\label{eq:pair-relabel}
	&=  \Pip^{\Frag'}_{a:i} \Pip^{\Grag'}_{b:j} \ket{\psi}.
	\end{split}\end{align}
	Similarly, we can swap the order:
	\begin{align}\begin{split}
	\Pip^\Frag_{a:i}\Pip^\Grag_{b:j}\ket{\psi} &= \Pip^{\Disj{\Frag}}_{a:i} \Pip^{\Grag'}_{b:j} \ket{\psi}\\
&=  \Pip^{\Grag'}_{b:j} \Pip^{\Disj{\Frag}}_{a:i}  \ket{\psi}\\
&=  \Pip^{\Grag'}_{b:j} \Pip^{\Frag'}_{a:i}  \ket{\psi}
	\end{split}\end{align}
	where the first line follows from \eqref{eq:pair-relabel} and we have used $\Disj{\Frag} = \Disj{\Frag}(\Grag,\Grag')$. 
	Note that this works when $\Frag=\Frag'$, $\Grag=\Grag'$, or both.
	
	Now assume we have shown that we can change the record locations ($\Grag\to \Grag'$, $\Hrag \to \Hrag'$, etc.)\ and do arbitrary operator re-orderings for some string of $M$ record projectors acting on the state:
	\begin{align}
	\label{eq:assump}
	\Pip^\Grag_{b:j}\cdots\,\Pip^{\Hrag}_{c:k}\ket{\psi}
	&=  \Pip^{\Grag'}_{b:j}\cdots\,\Pip^{\Hrag'}_{c:k}\ket{\psi}.
	\end{align}
	Then we must prove that if we left-multiply by an additional record projector $\Pip^\Frag_{a:i}$ then this property holds for the longer string of $M+1$ record projectors:
	\begin{align}
	\Pip^\Frag_{a:i} \left( \Pip^\Grag_{b:j}\cdots\,\Pip^{\Hrag}_{c:k} \right)\ket{\psi}
	&=  \Pip^{\Frag'}_{a:i} \left( \Pip^{\Grag'}_{b:j}\cdots\,\Pip^{\Hrag'}_{c:k} \right) \ket{\psi}.
	\end{align}
	To show that, we just apply the inductive assumption \eqref{eq:assump} with the choices $\Disj{\Grag} = \Disj{\Grag}(\Frag,\Frag')$, $\Disj{\Hrag} = \Disj{\Hrag}(\Frag,\Frag')$, etc.:
	\begin{align}\begin{split}
	\Pip^\Frag_{a:i} \left( \Pip^\Grag_{b:j} \cdots\, \Pip^{\Hrag}_{c:k} \right) \ket{\psi}
	&= \Pip^\Frag_{a:i} \left( \Pip^{\Disj{\Grag}}_{b:j} \cdots\, \Pip^{\Disj{\Hrag}}_{c:k} \right) \ket{\psi} \\
	&= \left( \Pip^{\Disj{\Grag}}_{b:j} \cdots\, \Pip^{\Disj{\Hrag}}_{c:k} \right)  \Pip^\Frag_{a:i} \ket{\psi} \\
	&= \left( \Pip^{\Disj{\Grag}}_{b:j} \cdots\, \Pip^{\Disj{\Hrag}}_{c:k} \right)  \Pip^{\Frag'}_{a:i} \ket{\psi} \\
	&= \Pip^{\Frag'}_{a:i} \left( \Pip^{\Disj{\Grag}}_{b:j} \cdots\, \Pip^{\Disj{\Hrag}}_{c:k} \right) \ket{\psi} \\
	\label{eq:mrelabel}
	&= \Pip^{\Frag'}_{a:i} \left( \Pip^{\Grag'}_{b:j} \cdots\, \Pip^{\Hrag'}_{c:k} \right)  \ket{\psi}.
	\end{split}\end{align}
	Similarly, it's easy to check that the $M+1$ record projectors can be arbitrarily re-ordered.  To do this, all we need to note is that the left-most two projectors can be commuted:
	\begin{align}\begin{split}
	\Pip^\Frag_{a:i} \Pip^\Grag_{b:j} \left( \Pip^{\Hrag}_{c:k} \cdots \right) \ket{\psi} &= \Pip^\Frag_{a:i} \Pip^{\Disj{\Grag}}_{b:j} \left( \Pip^{\Hrag}_{c:k} \cdots \right) \ket{\psi}\\
	&= \Pip^{\Disj{\Grag}}_{b:j} \Pip^\Frag_{a:i}  \left( \Pip^{\Hrag}_{c:k} \cdots \right) \ket{\psi}\\
	&= \Pip^{\Grag}_{b:j} \Pip^\Frag_{a:i}  \left( \Pip^{\Hrag}_{c:k} \cdots \right) \ket{\psi}
	\end{split}\end{align}
	where the first and last lines follows from \eqref{eq:mrelabel}.
	Commuting the two left-most operators, combined with the inductive assumptions, is sufficient to then re-order the entire product.
	
	By induction, this proves that we can relabel a product of any length, with arbitrary re-orderings.
	
	It's then easy to see that the recorded observables are compatible.  The branch decomposition is given by
	\begin{gather}
	\ket{\psi} = \sum_{i,j,k,\ldots} \ket{\psi_{i,j,k,\ldots}},\\
	\ket{\psi_{i,j,k,\ldots}} = 
	\left(\Pip^\Frag_{a:i}\Pip^\Grag_{b:j}\cdots \right)\ket{\psi}
	\end{gather}
	where the branches are unambiguously defined since they do not depend on the choices $\Frag$, $\Grag$, etc., nor the order of the projectors.  It  follows that each branch is an eigenvector of every record.
\end{proof}

\begin{cor}\corollaryDstatement\end{cor} 

\begin{proof}
	For any two of the observables, $\Obs_a$ and $\Obs_b$, consider the three records $\Grag$, $\Grag'$, and $\Grag''$ of $\Obs_b$ that are pairwise separated by the distance $\ell$.  Any pair of regions $\Frag$ and $\Frag'$ holding records of $\Obs_a$, which each have diameter at most $\ell$, can together overlap with at most two of $\Grag$, $\Grag'$, and $\Grag''$ by virtue of that triplet's spatial separation.  Therefore, $\Obs_a$ does not pair-cover $\Obs_b$, for all $a$ and $b$, so by the main theorem all observables are compatible.

	That this bound is tight is demonstrated by considering a Bell state, $\ket{\Phi^{+}} \propto \ket{\uparrow} \ket{\uparrow} + \ket{\downarrow}\ket{\downarrow} = \ket{\rarw}\ket{\rarw} + \ket{\larw} \ket{\larw}$. The observables $\Obs_\updownarrow$ and $\Obs_{\rarw,\larw}$ may each be recorded on arbitrarily small regions separated by an arbitrary long distance, yet are incompatible.
\end{proof}

\section{Appendix B: Insufficiency of pair-covering}
\label{sec:insuff}

Even when one recorded observable pair-covers another, it may still be possible to prove that the two observables are compatible using only the spatial layout of their records.  In other words, pair-covering, as a condition on spatial regions, is not sufficient to guarantee the existence of a wavefunction with incompatible observables recorded on those regions.  A simple counterexample can be constructed on a tripartite Hilbert space $\Glob = \mcA \otimes \mcB \otimes \mcC$ where $\Obs_a$ is recorded separately on all three regions and $\Obs_b$ on just $\mcA$ and $\mcB$.  Although $\Obs_a$ pair-covers $\Obs_b$, one can use the commutation of disjoint records to directly check that $\Pip_{a:i}^\Frag \Pip_{b:j}^\Grag \ket{\psi}$ is unchanged for any choice of $\Frag \in\{\mcA,\mcB,\mcC\}$ and $\Grag \in\{\mcA,\mcB\}$, and after swapping the order of $\Pip_{a:i}^\Frag$ and $\Pip_{b:j}^\Grag$.

\section{Appendix C: Shor code example}
\label{sec:shor}

Quantum codes are classes of quantum states used to encode quantum information in a noisy, correctable memory.  One of the important examples are the Shor codes, especially the special case of the 9-qubit code.  Here we show that the Shor codes are useful counterexamples when considering the mutual compatibility of recorded observables.  In particular, two incompatible observables can each be recorded with arbitrary redundancy using the Shor code.

For two integers $M,M'>1$, consider a many-body system with $MM'$ parts organized with the tensor structure 
\begin{align}
\Glob = \bigotimes_{m=1}^M \bigotimes_{m'=1}^{M'} \Sys^{(m,m')},
\end{align}
where each subsystem $\Sys^{(m,m')}$ has the same dimension $d$.  For our purposes, we will work with just qubits, $d=2$.

\begin{widetext}
	For an arbitrary single-qubit state $\ket{\psi} = \alpha\ket{0}+\beta\ket{1}$, let the isometry $\Lambda$ be defined to map $\ket{\psi}$ to the many-body state
	\begin{align}\begin{split}
	\ket{\Psi} &= \Lambda(\ket{\psi}) \\
	&=  \alpha \left[ \ket{0}^{\otimes M'} +  \ket{1}^{\otimes M'} \right]^{\otimes M} + \beta \left[ \ket{0}^{\otimes M'} -  \ket{1}^{\otimes M'} \right]^{\otimes M}\\
	&= \alpha \left[ \ket{00 \cdots 0} +  \ket{11 \cdots 1} \right]\otimes \cdots \otimes \left[ \ket{00 \cdots 0} +  \ket{11 \cdots 1} \right] 
	+ \beta \left[ \ket{00 \cdots 0} -  \ket{11 \cdots 1} \right]\otimes \cdots \otimes \left[ \ket{00 \cdots 0} -  \ket{11 \cdots 1} \right]
	\end{split}\end{align}
	It can be shown that this state is a redundant encoding of two complementary observables, $\sigma_z = \projector{0}-\projector{1}$ and $\sigma_x = \ketbra{0}{1} + \ketbra{1}{0}$, for the original state $\ket{\psi}$.

	This is instructive to see in the special case of $M=M'=3$, for which the state reduces to the 9-qubit code. 
	For clarity, we label the subsystems $A$, $B$, $C$, etc.:
	\begin{align}\begin{split}
	\ket{\Psi} = \alpha &\left[ \ket{000}_{ABC} +  \ket{111}_{ABC} \right]\otimes  \left[ \ket{000}_{DEF} +  \ket{111}_{DEF} \right]\otimes \left[ \ket{000}_{GHI} +  \ket{111}_{GHI} \right] \\
	+ \beta &\left[ \ket{000}_{ABC} -  \ket{111}_{ABC} \right]\otimes  \left[ \ket{000}_{DEF} -  \ket{111}_{DEF} \right]\otimes \left[ \ket{000}_{GHI} -  \ket{111}_{GHI} \right]
	\end{split}\end{align}
	It's clear that by making a measurement on just $ABC$ in the basis $\{\ket{000}+\ket{111},\ket{000}-\ket{111}\}$ we get the outcomes with respective probabilities $\abs{\alpha}^2$ and $\abs{\beta}^2$, and that this provides as much information as measuring the original state $\ket{\psi}$ in the computational basis $\{\ket{0},\ket{1}\}$, i.e., measuring the observable $\sigma_z$.  Likewise is true for observers who have access to only $DEF$ or $GHI$.  
	
	Now we can rewrite this code state as 
	\begin{align}\begin{split}
	\ket{\Psi} = (\alpha +\beta) \big[ &\ket{000}_{ADG} \ket{000}_{BEH}\ket{000}_{CFI} +\ket{011}_{ADG} \ket{011}_{BEH}\ket{011}_{CFI} \\
	+ &\ket{101}_{ADG} \ket{101}_{BEH}\ket{101}_{CFI} + \ket{110}_{ADG} \ket{110}_{BEH}\ket{110}_{CFI} \big] \\
	+ (\alpha -\beta) \big[ &\ket{001}_{ADG} \ket{001}_{BEH}\ket{001}_{CFI} +\ket{010}_{ADG} \ket{010}_{BEH}\ket{010}_{CFI} \\
	+ &\ket{100}_{ADG} \ket{100}_{BEH}\ket{100}_{CFI} + \ket{111}_{ADG} \ket{111}_{BEH}\ket{111}_{CFI} \big]
	\end{split}\end{align}
	and consider an alternate observer who instead measures the subset $ADG$ in the computational basis $\{\ket{000},\ket{001},\ldots\}$, or any other basis that distinguishes whether there is an even or odd number of 1's.  The outcome will be an even or odd number of 1's with respective probabilities $\abs{\alpha+\beta}^2$ and $\abs{\alpha-\beta}^2$.  This provides as much information as measuring the original state $\ket{\psi}$ in the complementary basis $\{\ket{0}+\ket{1}, \ket{0}-\ket{1}\}$, i.e., measuring the observable $\sigma_x$. Likewise is true for observers who have access to only $BEH$ or $CFI$.
\end{widetext}

In fact, there is a symmetry between $\sigma_z$ and $\sigma_x$, so it's also possible to measure $\sigma_z$ using only single-qubit measurements on one of the record regions.  One can distinguish $\ket{000}_{ABC}+\ket{111}_{ABC}$ from $\ket{000}_{ABC}-\ket{111}_{ABC}$ by measuring each qubit in the basis $\ket{\pm} \propto \ket{0}\pm \ket{1}$ and taking the parity. Of course, if $\sigma_z$ is first determined by measuring $ABC$, then the information about $\sigma_x$ is destroyed and cannot be obtained through a measurement on $ADG$, and vice versa.

These measurements work just as well for arbitrary $M,M'>3$.  Therefore, if we have $M$ observers and have them each make a measurement on $\Frag^{[m]} = \bigotimes_{m=1}^{M'} \Sys^{(m,m')}$, then $\sigma_z$ is recorded redundantly ($M$-fold times), with a local copy at each observer.  On the other hand, we could have $M'$ observers each make a measurement on $\Grag^{[m']} = \bigotimes_{m=1}^M \Sys^{(m,m')}$, 
so that $\sigma_x$ is recorded redundantly ($M'$-fold times), with a local copy at each observer.

\end{document}